\documentclass[letterpaper, 10 pt, conference]{ieeeconf}
\usepackage{graphicx,mathtools,algorithm,algorithmic,amsmath,wrapfig,subfigure,amssymb,multirow,cite}
\let\labelindent\relax
\usepackage{enumitem}
\usepackage{epstopdf}
\newtheorem{mydef}{Definition}
\newtheorem{thrm}{Theorem} \newtheorem{remark}{Remark} \newtheorem{lem}{Lemma}
\newtheorem{aus}{Assumption}
\newcommand{\norm}[1]{\left\lVert#1\right\rVert}
\usepackage{booktabs} \usepackage[usenames, dvipsnames]{color}
\usepackage{pbox} 
\usepackage{epstopdf} 

\IEEEoverridecommandlockouts
\title{Robust Stabilization of Resource Limited Networked Control Systems Under Denial-of-Service Attack}
\author{Niladri Sekhar Tripathy, Mohammadreza Chamanbaz and Roland Bouffanais%
\thanks{Niladri Sekhar Tripathy, Mohammadreza Chamanbaz and Roland Bouffanais are with the Singapore University of Technology and Design, Singapore 487372, e-mail: ({\tt\footnotesize niladri.tripathy, chamanbaz, bouffanais@sutd.edu.sg})}%
\thanks{
	This work was supported in part by the National Research Foundation (NRF), Prime Minister’s Office, Singapore, under its National Cybersecurity R\&D Programme (Award No. NRF2014NCR-NCR001-040) and administered by the National Cybersecurity R\&D Directorate and MOE Tier 1 grant (T1MOE17001).}
}
  
\begin{document}

\maketitle


\begin{abstract}
In this paper, we consider a class of  denial-of-service (DoS) attacks, which aims at overloading the communication channel. On top of the security issue, continuous or periodic transmission of information within feedback loop is necessary for the effective control and stabilization of the system. In addition, uncertainty---originating from variation of parameters or unmodeled system dynamics---plays a key role in the system's stability. To address these three critical factors, we solve the joint control and security problem for an uncertain discrete-time Networked Control System (NCS) subject to limited availability of the shared communication channel. An event-triggered-based control and communication strategy is adopted to reduce bandwidth consumption. To tackle the uncertainty in the system dynamics, a robust control law is derived using an optimal control approach based on a virtual nominal dynamics associated with a quadratic cost-functional. The conditions for closed-loop stability and aperiodic transmission rule of feedback information are derived using the discrete-time Input-to-State Stability theory. We show that the proposed control approach withstands a general class of DoS attacks, and the stability analysis rests upon the characteristics of the attack signal. The results are illustrated and validated numerically with a classical NCS batch reactor system.
\end{abstract}
%

\section{Introduction}
%
The range of applications of Cyber-Physical Systems (CPSs)---e.g. power systems, intelligent
vehicles, civil infrastructure, aerospace, retail supply chains, connected medical devices---has vastly expanded beyond the realm of large-scale public infrastructures. The presence of a communication medium combined with a tight integration of various subsystems make most of these applications safety-critical. Therefore, both CPSs and Networked Control Systems (NCSs) are broadly exposed to cyber-threats and cyber-vulnerabilities which may affect the functionality of physical processes at their core. These critical issues have spurred new lines
of research at the interface between cyber-security and control
theory~\cite{dosx1,dosx3}. For instance, the effects and containment of
cyber-attacks on control systems, which affect the availability and integrity
of sensor and actuator information have been studied
in~\cite{dosx4,dosx5}. Recently, Teixira et. al.~\cite{dosx5} described
different characteristics of cyber-attacks and defined an attack space to
analyze the effect of cyber-attacks on closed-loop dynamics.  Cyber-attacks
can be broadly classified into two categories: Denial-of-Service (DoS) attacks
and deception attacks~\cite{se1z}. This paper is concerned with DoS attacks and
their effects on dynamical systems. DoS attacks primarily affect the transmission
medium within the feedback loop and cause irregular exchanges and losses of
information~\cite{dosx6,dosx7}. As this is one of the most reachable attack
patterns in the attack space, many researchers have studied its effects both
theoretically and experimentally~\cite{dosx8,lshi1z,lshi2z,cd123}.

\par
Beyond inherent security issues present in NCSs, the exchange of feedback
information over the shared communication channel, be it continuous or
periodic, consumes a significant portion of the available bandwidth.
Recently, it has been shown that significant savings in the bandwidth and
communication resources can be achieved by switching from periodic or
continuous sampling to aperiodic sampling~\cite{nst5,nst31}.
Specifically, event-triggered control strategies have revealed drastic
reductions in the use of network bandwidth within the feedback
loop~\cite{something6,new4,nn1,nstzc1,nstcd}.  
A central problem with classical event-triggered control is the need to have
an accurate model of the system in order to devise appropriate
event-triggering rules. In practice, system modeling inevitably simplifies the
actual system's operations, and thereby introduces a certain level of
inaccuracy.
Recently, Tripathy et al.~\cite{nstzc1} have developed a
robust event-triggered control algorithm based on aperiodic feedback so as to
deal with the presence of uncertainty. 
%
\par
It is worth highlighting that there is a vast breadth of problems related to
the issue of event-triggering control in the presence of DoS attacks, and with
model uncertainty in NCSs.
In event-triggered control, any new information is exchanged only when the
stability criterion is violated, which implicitly assumes that the
communication channel is available at the time of event generation. It is
clear that any factor or event affecting the availability of the
interconnecting network, such as a DoS attack for instance, has the potential
to seriously hinder the underlying physical processes and overall operations
of the NCS. In light of this, it appears timely to develop new
event-triggering control strategies capable of ensuring the stability of the
closed loop system subjected to DoS attacks characterized by their frequency
and duration, while accounting for uncertainty of the NCS model.

\par
In this paper, we propose an attack-resilient event-based robust control
algorithm for discrete-time uncertain systems. Norm-bounded mismatched
uncertainty is considered for the derivation of the robust control
results. 
The primary goal of this
work is to analyze the effect of DoS attacks on a discrete-time uncertain
network controlled system, and to characterize the relationship between
frequency and duration of the attack signal and closed-loop stability. The
Input-to-State Stability (ISS) theory is applied to derive the transmission
rule and on/off periods of DoS attack signal. The key contributions of this
paper are listed below:
\begin{itemize}[leftmargin=.25cm,labelsep=0cm]
\item We derive and propose a resilient event-based robust control law, within
  the optimal control framework, that is capable of dealing with both the
  occurrence of repeated DoS attacks and model uncertainty.
\item We establish the upper bound of acceptable duration and frequency of DoS
  attacks, for which the ISS stability of the uncertain discrete-time systems
  is guaranteed with event-triggered feedback.
\item The numerical results obtained with a NCS model for a batch reactor system provide an
  illustration of the proposed approach and also validate its effectiveness.
\end{itemize}
\par

%
\subsubsection*{Notations and Definitions}

The Euclidean norm of a vector $x\in{\mathbb{R}^{n}}$ is denoted by $\|x\|$. 
The symbols $I$ 
denote the identity matrix of appropriate dimension. The maximum
(resp. minimum) eigenvalue of a symmetric matrix $P\in \mathbb{R}^{n\times n}
$ is $\lambda_{\text{max}}(P)$ (resp. $\lambda_{\text{min}}(P)$). 
A continuous
function $f$: $\mathbb{R}_{\geq 0}\rightarrow \mathbb{R}_{\geq 0}$ is said to
be class $\mathcal{K}_{\infty}$ if it is strictly increasing, $f(0)=0$ and
$f(s)\rightarrow\infty$ as $s\rightarrow\infty$. A  function
$f:\mathbb{R}_{\geq 0}\rightarrow \mathbb{R}_{\geq 0}$ is a class
$\mathcal{K}$ function, if it is continuous, strictly increasing and
$f(0)=0$. A continuous function $\beta(r,s): \mathbb{R}_{\geq 0}\times
\mathbb{R}_{\geq 0}\rightarrow \mathbb{R}_{\geq 0}$ is a $\mathcal{KL}$
function, if it is a class $\mathcal{K}$ function with respect to $r$ for a
fixed $s$, and it is strictly decreasing with respect to $s$ when $r$ is
fixed~\cite{nst69}. For any given time interval $[0, k)$ where
$k>1$, $T_{\text{off}}(k)$ denotes the total duration of DoS attack over $[0,k)$. The
ratio $\frac{T_{\mathrm{off}}(k)}{k}$ represents the rate of unavailability of
the communication channel following the DoS attack.  The variable
$N_{\mathrm{off}}(k)$ represents the frequency of DoS attack in the time
interval $[0,k)$ i.e. it means that $N_{\mathrm{off}}(k)$ off-to-on
transitions are present in the attack signal during which communication is
impossible. The definitions detailed below are used to
establish the theoretical results.
\begin{mydef}[Input-to-State Stability\cite{nst69}]\label{def23}~\\
  A discrete-time system
  \begin{eqnarray}\label{mbcz45}
    x(k+1)=Ax(k)+Bu(k),
  \end{eqnarray}
  is globally input-to-state stable (ISS) if it satisfies
  \begin{eqnarray}
    \|x(k)\|\leq\beta(\|x(0)\|,k)+\gamma\left( \|u(k)\|\right),
  \end{eqnarray}
  for all admissible inputs $u(k)$ and for all initial values $x(0)$, with
  $\beta$ a $\mathcal{KL}$ function, and $\gamma$ a $\mathcal{K}_{\infty}$
  one.
\end{mydef}
\begin{mydef}[ISS Lyapunov Function \cite{nst69}]\label{def34}~\\
  Assume system \eqref{mbcz45} is at steady state at the origin, that is
  $f(0,0)=0$, $\forall \ k>0$. A positive function $ V(x(k)):
  \mathbb{R}^{n}\rightarrow \mathbb{R}$ is an Input-to-State Lyapunov function
  for~\eqref{mbcz45} if there exists class $\mathcal{K}_{\infty}$ functions $
  \alpha_{1}, \alpha_{2}, \alpha_{3}$ and a class $\mathcal{K}$ function $
  \gamma$ for all $x\in\mathbb{R}^{n}$ and $ u \in\mathbb{R}^{m}$ satisfying
  the following conditions
  \begin{flalign}\label{deft1}
    &\alpha_{1}(\|x(k)\|)\leq V (x(k))\leq \alpha_{2}(\|x(k)\|),\\&
    {V}(k+1)-V(k)\leq -\alpha_{3}{(\|x(k)\|)}+\gamma{(\|u(k)\|)}.\label{deft2}
  \end{flalign}
\end{mydef}

%
\section{ Problem Formulation and Preliminaries}\label{cbf}
\subsection{Problem description}
Consider a linear event-triggered system with model uncertainty mathematically
represented by
\begin{align}\label{dosa1}
  x(k+1)=&\big(A+\Delta A(p)\big)x(k)+Bu(k_{i}),\nonumber\\ & \qquad\forall~k\in [k_{i}, k_{i+1}), \ i\in \mathbb{N},\\
  u(k_{i})=&Kx(k_{i})=K\{x(k)+e(k)\},\label{kma1} 
\end{align}
where $x(k)\in \mathbb{R}^{n}$ and $u(k_{i})\in \mathbb{R}^{m}$ are the system
state and input vectors, respectively. The symbol $k_i$ in \eqref{dosa1} and
\eqref{kma1} represents the $i$-th aperiodic sensing and actuation instant  and $  e(k)=x(k_{i})-x(k), \forall k\in [k_{i}, k_{i+1}).
$ 
The unknown matrix
$\Delta A(p)\in \mathbb{R}^{n\times n}$ represents the uncertainty due
to the bounded variations of the system's parameter $p$ and its effects on the
nominal system matrix $A$. The variations of $p$ are bounded by a known and
possibly uncountable set $\Omega$. In general, the uncertainty is either
matched or mismatched~\cite{fv1}. For matched system, the uncertainty affects
the system's dynamics via the input matrix, i.e. $\Delta A(p)$ is in
the range space of matrix $B$. This assumption does not hold for mismatched
systems. In this paper, the unknown matrix $\Delta A(p)$ is mismatched
in nature and it is expressed as
\begin{flalign}\label{dosh56}
  \Delta A(p)
  =\underbrace{BB^{+}\Delta A(p)}_{\text{matched}}+\underbrace{(I-BB^{+})\Delta A(p)}_{\text{mismatched}}.
\end{flalign}
The
matrix $B^{+}$ represents the left-pseudo inverse of input matrix $B$,
i.e. $B^{+}=(B^{T}B)^{-1}B^{T}$. The unknown state perturbation
matrix $\Delta A(p)$ is bounded by a known matrix $F$ which is defined as 
\begin{eqnarray}\label{zok1}
  \Delta A(p)^{T}\Delta A(p)\leq \epsilon\frac{F}{2},
\end{eqnarray}
where the scalar $\epsilon$ is a design parameter.
\begin{figure}
  \begin{center}
    \includegraphics[width=1\columnwidth]{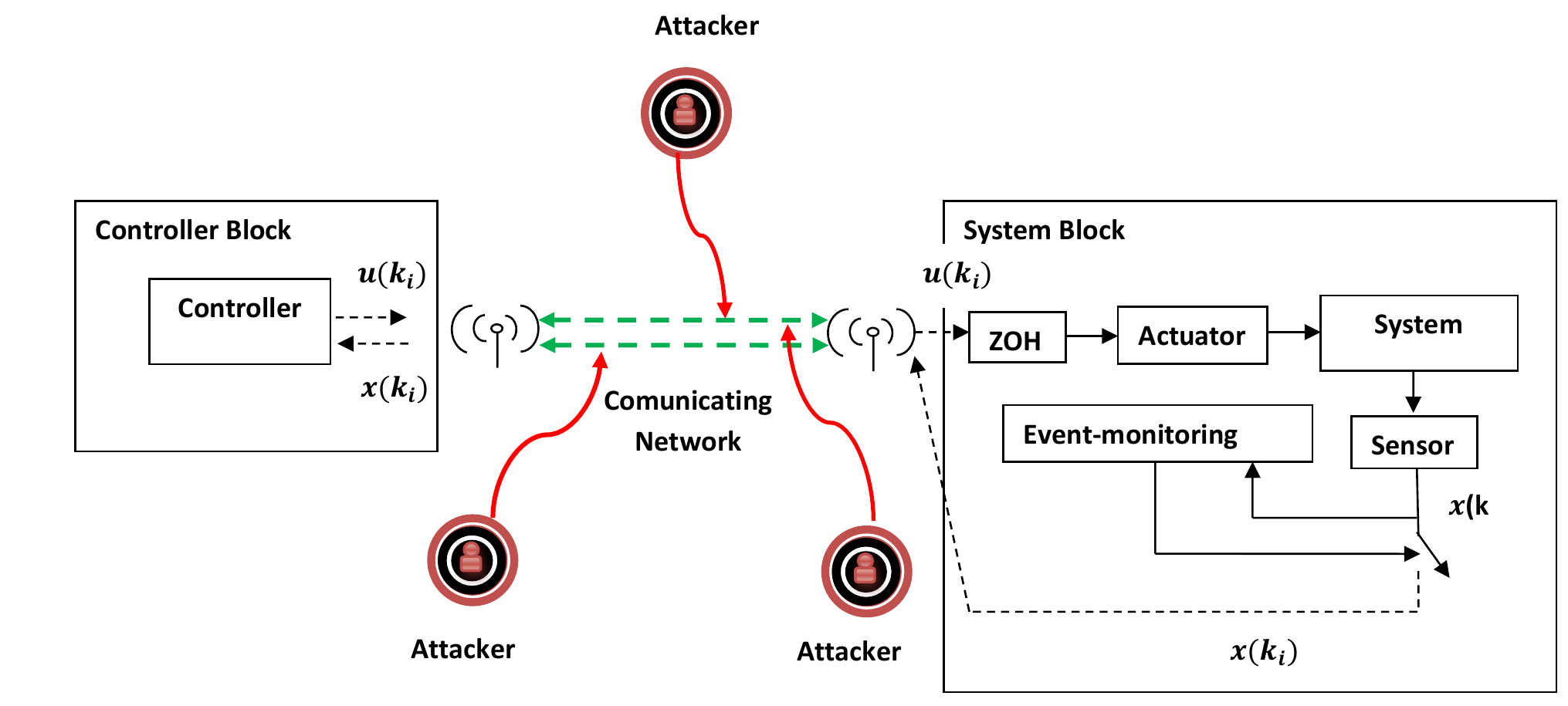}
    \caption{ Block diagram of proposed control technique under DoS
      attack}\label{fi:m1}
  \end{center}
\end{figure}
The block diagram of the proposed controlled system is shown in
Fig. \ref{fi:m1}.
According to~\eqref{dosa1}, the control and sensing
actions are executed at each event-triggering instant $k_i$. However, when DoS
interruptions affect the communication medium, the control and sensing actions
are prevented from being executed. For simplicity, in this paper we 
assume that DoS attack equally affects the control and measurement
channels. As expected, in the presence of DoS attacks, the data cannot be
transmitted to or received from the communication channel.\\
\textbf{Problem Statement:} Design robust event-triggered state feedback
control law \eqref{kma1} that stabilizes  system~\eqref{dosa1} in the
presence of DoS attacks  and  mismatched uncertainty~\eqref{dosh56}. \\
\textbf{Proposed Solution:} A two-step solution to this control problem is
proposed. First, a robust controller is designed to handle uncertainty and
then, a transmission rule for sensing and actuation is derived to tackle DoS
effects and limited availability of communication channel. To derive the
robust controller gain, an emulation-based approach is adopted from~
\cite{nstzc1}. That means, the controller is designed excluding the influence
of the network, and then some conditions are derived to deal with network
constraints. In~\cite{nstzc1}, Tripathy et. al. derived the robust controller
gain matrices within the optimal control framework, which is discussed next.

\subsection{Optimal Control Approach for Robust Controller Design} The optimal
control solution for a virtual system
\begin{flalign}\label{nomi1}
  x(k+1)= & Ax(k)+Bu(k)+\alpha(I-BB^{+})v(k),
\end{flalign}
which minimizes a modified cost function
\begin{flalign}
  J(k)= &\frac{1}{2}\sum_{k=0}^{\infty}
  \bigg\{x(k)^{T}(Q+F)x(k)+u(k)^{T}R_{1}u(k)\nonumber\\
  &+v(k)^{T}R_{2}v(k)\bigg\},\label{cos1}
\end{flalign}
is robust for the original systems~\eqref{dosa1} in the presence of uncertainty defined in
\eqref{dosh56}. Here, $\alpha$ is a scalar and $Q\geq 0$, $R_1>0$, $R_2>0$ are matrices. The system~\eqref{nomi1} has
two control inputs $u$ and $v$, which are denoted as stabilizing and virtual
inputs respectively. The importance of virtual input $v$ is discussed in
Remark \ref{dosr1}. 
To design the robust controller gains for \eqref{dosa1}, the optimal control
problem for \eqref{nomi1} and \eqref{cos1} is solved adopting the method
proposed in \cite{nstzc1,rnc001} and results are presented as a Lemma below.
\begin{lem}\label{doslem1}
  Suppose there exist a scalar $\epsilon >0$ and positive definite solution
  $P>0$ of the following Riccati equation
  \begin{eqnarray}\label{ri1}
    & A^{T}\big \{P^{-1}+BR_{1}^{-1}B^{T}+\alpha^{2}(I-BB^{+})R_{2}^{-1}\nonumber\\ &(I-BB^{+})^{T}\big\}^{-1}A-P+Q+F=0,
  \end{eqnarray}
  and 
  \begin{eqnarray}\label{scon1}
    (\epsilon^{-1}I-P)>0.
  \end{eqnarray}
  If the optimal control inputs $u=Kx$ and $v=Lx$ for \eqref{nomi1} and
  \eqref{cos1} are selected as
  \begin{flalign}\label{g1}
    K =&-R_{1}^{-1}B^{T}\big\{P^{-1}+BR_{1}^{-1}B^{T}\nonumber \\ &+\alpha^{2}(I-BB^{+})R_{2}^{-1}(I-BB^{+})^{T}\big\}^{-1}A, \\
    L=&-\alpha
    R_{2}^{-1}(I-BB^{+})^{T}\big\{P^{-1}+B^{T}R_{1}^{-1}B\nonumber \\
    &+\alpha
    ^{2}(I-BB^{+})R_{2}^{-1}(I-BB^{+})^{T}\big\}^{-1}A,\label{g2}
  \end{flalign}
  where the gain matrices $K$ and $L$ satisfy the following matrix inequality
  \begin{eqnarray}\label{scon3}
   Q_{1}= (Q+K^{T}R_{1}K+L^{T}R_{2}L+M^{T}P^{-1}M)\nonumber \\ - A_{c}^{T}\big(P^{-1}-\epsilon I\big)^{-1}A_{c}> 0,
  \end{eqnarray} 
  with $A_{c}=A+BK$ and
  \begin{align}\label{kma3}
    M&= \{P^{-1}+BR_{1}^{-1}B^{T}\nonumber \\ &+\alpha^{2}(I-BB^{+})R_{2}^{-1}(I-BB^{+})^{T}\}^{-1}A,
  \end{align}
  then, the matrix $K$ is the robust controller gain for \eqref{dosa1}.
\end{lem}
\vspace{0.1in}
The detailed proof of Lemma~\ref{doslem1} can be found in~\cite{nstzc1}. In
the next section, the gain matrices $K$ and $L$ are used to derive the
transmission instant.
\begin{remark}\label{dosr1}
  The virtual system \eqref{nomi1} has two control inputs $u=Kx$ and
  $v=Lx$. The virtual input $v$ is used for handling the mismatched
  uncertainty even though $v$ is not used directly to stabilize the uncertain
  system \eqref{dosa1}. However, $v$ indirectly helps to design the robust
  controller gain $K$ by satisfying the inequality \eqref{scon3}.
\end{remark}
\section{Main results}
%
In this section, we consider a class of DoS attacks and present an event-triggering rule robustly stabilizing the closed loop system in the presence of model uncertainty and DoS attack. In particular, we assume that the DoS attack holds the following assumptions. 
\begin{aus}\label{a1s}[DoS attack rate] There exist scalars
    $\eta_{1}, c_{1}, c_{2}\in \mathbb{R}$  such that 
    \begin{eqnarray}\label{off1}
      \frac{T_{\text{off}}(k)}{k}\leq \frac{2\ln(\eta_{1})-\ln(c_{1})}{\ln(c_{2})-\ln(c_{1})},\qquad\forall~k>1. 
    \end{eqnarray}
   where $c_1<\eta_{1}<1$ and $c_2>1$.
  \end{aus}
  \vspace{1ex}
  \begin{aus}\label{a2s}[DoS frequency] Let $T_a$ be  the average time between two consecutive attacks and  suppose scalar $\eta_2$ satisfies $1>\eta_{2}>\eta_{1}$. Then, the frequency of DoS attack for an interval $[0,k)$ is upper bounded by
    \begin{eqnarray}\label{nqdos26}
      \frac{N_{\mathrm{off}}(k)}{k}\leq{T_a},
    \end{eqnarray}  
    where $ T_a= \frac{2(\ln(\eta_{2})-\ln(\eta_{1}))}{\ln\left(\lambda_{\max}(P)/\lambda_{\min}(P)\right)} $.
  \end{aus}
\vspace{1ex}
Assumptions \ref{a1s} and \ref{a2s} imply some restrictions on the nature of
the DoS attack in terms of duration and frequency of attack. For the sake of
the analysis, we limit our study to the class of DoS signals satisfying both
Assumptions \ref{a1s} and \ref{a2s}. Owing to the occurrence of DoS attack disrupting
 the communication channel, the transmission of information at time instant
$k_i$ may be influenced. 

To prove the stability of the closed-loop system \eqref{dosa1} and to design
an event-triggering rule that can withstand model uncertainty in the presence
of DoS attacks, the following two cases are considered. First, we establish
the stability results and derive an event-triggering condition in the absence
of any DoS attack. Second, to circumvent the DoS-related effects, we derive
some conditions that the attack signal must satisfy for our event-triggering
approach to be effective. Before stating the main theorem, the following two
lemmas adopted from \cite{nstzc1,gar1} are introduced which are instrumental
to prove the main results.
\begin{lem}\label{lemdos1}
  Suppose there exists a positive definite solution $P>0$ of \eqref{ri1} and a
  scalar $\epsilon>0$. Then if $(\epsilon^{-1}I-P)>0$, the following holds
  \small
  \begin{equation}
    \hat{X}^{T}P\hat{W}+\hat{W}^{T}P\hat{X}+\hat{W}^{T}P\hat{W}\leq  \hat{X}^{T}(\epsilon^{-1}I-P)^{-1}\hat{X} +\epsilon^{-1}\hat{W}^{T}\hat{W},
  \end{equation}
  \normalsize
  where $\hat{X}$ and $\hat{W}$ are two matrices with appropriate dimensions.
\end{lem}
\begin{lem}\label{lemdos2}
  Let $P>0$ be a solution of \eqref{ri1} and the gain matrices $K$ and $L$ be
  computed using \eqref{g1} and \eqref{g2}, respectively. Using \eqref{g1} and
  \eqref{g2} the following holds
  \begin{flalign}
    & A^{T}(P^{-1}+BR^{-1}B^{T}+\alpha^{2}(I-BB^{+})R_{2}^{-1} (I-\nonumber
    \\ & BB^{+})^{T})^{-1}A= K^{T}R_{1}K+L^{T}R_{2}^{T}L+M^{T}P^{-1}M,
  \end{flalign}
  where matrix
  $M$ is defined in \eqref{kma3}.
\end{lem}
 
The main results of this paper are stated in the following theorem.

\begin{thrm}\label{nqdos23}
  Suppose there exist  
  scalars  $\sigma\in (0,1)$ and $\epsilon>0$ which satisfy \eqref{zok1} and \eqref{scon1}
  and let the controller gain matrices derived from \eqref{g1} and
  \eqref{g2}. Consider any DoS signal for which Assumptions \ref{a1s} $\&$
  \ref{a2s} hold. If~\eqref{scon3} holds and the control input \eqref{kma1} is
  actuated based on the following event-triggering sequence 
  \small
  \begin{equation}\label{nqdos28}
    k_{0}=0, k_{i+1}=\text{\em inf}\Big\{k\in \mathbb{N} |k\geq k_i\wedge(\mu\|x\|^{2}-\|e\|^{2})\leq 0\Big\},
  \end{equation}
  \normalsize
  with 
  \small
  \begin{equation}
    \mu=\frac{\sigma\lambda_{\text{min}}^{2}(Q_{1})}{4\|(A_{c}^{T}PBK)\|^{2}+2\lambda_{\min}(Q_{1})\|K^{T}B^{T}(P^{-1}-\epsilon I)BK\|},
  \end{equation}
  \normalsize
  then, the event-triggered control law \eqref{kma1} ensures the ISS of the
  system \eqref{dosa1} in the presence of uncertainty \eqref{dosh56} and DoS
  attacks.
\end{thrm}

The proof of Theorem~\ref{nqdos23} is divided into two cases discussed below.
\begin{description}
\item[Case 1.]\textbf{\ No DoS attack has occurred:} Here, we assume that the
  communication medium is perfect for data transmission, without any jamming
  within the channel. Therefore, any attempts in updating the control inputs
  will be successful. That means, whenever an event is generated, the
  transmission of sensor and control information are not interrupted and the
  control law is actuated immediately. The stability criteria and aperiodic
  transmission rule of information in the absence of any DoS attack are
  reported below for this particular case.
\item[Case 2.] \ \textbf{\ A DoS attack has occurred:} Here, we suppose that
  the attacker successfully compromises the effectiveness of the communication
  medium, thereby preventing feedback loops from operating from time to
  time. If the channel is not available to update the control actions, it may
  affect the closed-loop stability and sensing and actuation instants. In this
  case, we study the effect of attacks and model uncertainty in system's
  stability and propose a criterion guaranteeing the stability of the
  closed-loop system in the presence of DoS attacks satisfying Assumptions
  \ref{a1s} and \ref{a2s}.
\end{description}
\begin{proof}[Proof of Theorem \ref{nqdos23}]
  \textbf{Case 1.}  Let there exists an ISS Lyapunov function
  $V(k)=x^{T}Px$. 
  Using~\eqref{dosa1}, $\Delta V =[V(k+1)-V(k)]$ is computed
  as
  \begin{flalign*}
    \Delta V&= x^{T}[A_{c}^{T}PA_{c}+A_{c}^{T}P\Delta A+\Delta
    A^{T}PA_{c}+\Delta A^{T}P\Delta A]x \nonumber \\
    &+x^{T}A_{c}^{T}PBKe+x^{T}\Delta A^{T}PBKe+e^{T}K^{T}B^{T}PA_{c}x\nonumber
    \\ & +e^{T}K^{T}B^{T}P\Delta A x+e^{T}K^{T}B^{T}PBKe-x^{T}Px,
  \end{flalign*}
  where $A_{c}=A+BK$.  The above equality is simplified using Lemma
  \ref{lemdos1} as
  \begin{eqnarray}
    \Delta  V &\leq &  x^{T}[A_{c}^{T}(P+P(\epsilon^{-1}I-P)^{-1}P)A_{c}-P\nonumber\\ &&+2\epsilon^{-1}\Delta A^{T}\Delta A ]x+x^{T}A_{c}^{T}PBKe \nonumber\\ &&+e^{T}K^{T}B^{T}PA_{c}x+e^{T}K^{T}B^{T}(P\nonumber\\ &&+P(\epsilon^{-1}I-P)^{-1}P)BKe.
    \label{nodos2}
  \end{eqnarray}
  Using matrix inversion lemma and solution of Riccati equation from \eqref{ri1}, inequality \eqref{nodos2} is simplified as
    \begin{align}
    \Delta  V &\leq   x^{T}[A_{c}^{T}(P^{-1}-\epsilon I)^{-1}A_{c}-(Q+F)-A^{T}(P^{-1}\nonumber\\ &+BR^{-1}B^{T}+\alpha^{2}(I-BB^{+})R_{2}^{-1} (I- BB^{+})^{T})^{-1}A\nonumber\\ 
    &+2\epsilon^{-1}\Delta A^{T}\Delta A ]x+x^{T}A_{c}^{T}PBKe+e^{T}K^{T}B^{T}PA_{c}x \nonumber\\ &+e^{T}K^{T}B^{T}(P^{-1}-\epsilon I)^{-1})BKe.
    \label{nodosdf2}
  \end{align}
  {Using \eqref{zok1} and applying Lemma \ref{lemdos2} to \eqref{nodosdf2}, we arrive at
  \begin{align*}
    \Delta V & \leq x^{T}[A_{c}^{T}(P^{-1}-\epsilon I)^{-1}A_{c}-Q-K^{T}R_{1}K-L^{T}R_{2}L\nonumber\\ &-M^{T}P^{-1}M]x+\psi x^{T}x+\frac{1}{\psi}\|A_{c}^{T}PBK\|^{2}\|e\|^{2}\nonumber\\ & +e^{T}K^{T}B^{T}(P^{-1}-\epsilon I)^{-1})BKe,
  \end{align*}
  where $\psi$ is a positive scalar. Furthermore, using \eqref{scon3}, we can simplify above inequality to 
  \begin{align*}
      \Delta V&\leq -x^TQ_1x+\psi x^Tx+\frac{1}{\psi}\|A_{c}^{T}PBK\|^{2}\|e\|^{2}\nonumber\\ & +\|K^{T}B^{T}(P^{-1}-\epsilon I)^{-1})BK\|\|e\|^2.
  \end{align*}
  Choosing $\psi=\frac{\lambda_{\min}(Q_{1})}{2}$, the following
  is obtained
  \begin{eqnarray}\label{nodos3}
    \Delta V &\leq & -\xi_{1}\|x(k)\|^{2}+\xi_{2}\|e(k)\|^{2},
  \end{eqnarray}
  where matrix $Q_{1}$ is defined in \eqref{scon3} and  $\xi_{1}=\frac{\lambda_{\text{min}}(Q_{1})}{2}$ and $\xi_{2}=\left(
    \frac{2(\|A_{c}^{T}PBK\|^{2}}{\lambda_{\text{min}}(Q_{1})}+\|K^{T}B^{T}(P^{-1}-\epsilon I )^{-1}BK\|\right)$.} Using
  Definitions \ref{def23} and \ref{def34}, the inequality \eqref{nodos3}
  ensures the ISS of \eqref{dosa1}. In the absence of any DoS attack, the
  event-triggering condition \eqref{nqdos28} is also derived using
  \eqref{nodos3}.  In fact, the control inputs need to be actuated whenever
  the condition \eqref{nqdos28} is violated. \par
  The Lyapunov function $V(x)=x^{T}Px$ satisfies \eqref{deft1} where
  $\alpha_{1}(\norm{x})=\lambda_{\text{min}}(P)\|x\|^{2}$ and
  $\alpha_{2}(\norm{x})=\lambda_{\text{max}}(P)\|x\|^{2}$. Now applying the
  event-triggering condition \eqref{nqdos28}, the bound of $\Delta V$ can be
  written as
  \small
  \begin{equation}
    \Delta V(x) \leq   -\frac{\xi_{1}}{\lambda_{\text{min}}(P)}(1-\sigma)V(x)
     \leq -\frac{\lambda_{\text{min}}(Q_{1})}{2\lambda_{\text{min}}(P)}(1-\sigma)V(x),\label{xbx}
  \end{equation}
  \normalsize
  where  $\sigma\in (0,1)$ regulates the transmission of
  information over the network. { The information exchange over the network has inverse relation with the selection of $\sigma$.}  This proves that the closed-loop system
  \eqref{dosa1} is globally asymptotically stable with event-triggered
  feedback and model uncertainty. 
  Using \eqref{xbx}, following yields
  \begin{eqnarray}\label{vb1d}
    V(k+1)\leq c_{1}V(k)
  \end{eqnarray}
  where  $c_{1}=(1-\frac{\lambda_{\text{min}}(Q_{1})}{2\lambda_{\text{min}}(P)}(1-\sigma))$ and  is always less than $1$ as $V(x)$ is decreasing. 
  The following Remark describes the growth of
  error $e$ in between two consecutive events.


\begin{remark} Inequality \eqref{xbx} signifies that the state
      of the uncertain system \eqref{dosa1} will remain bounded. Since the
      state is bounded, the measurement error $e(k)$ is also bounded.  Here,
      the variable $e(k)$ evolves based on the following difference equation
      \begin{align}\label{nqdos1}
        &e(k+1)=x(k_{i})-x(k+1), \\ 
        &= (A+BK+\Delta A)e(k)+(I-(A+\Delta A+BK))x(k_{i})\nonumber.
      \end{align}
      The matrix $\Delta A $ is also bounded as the condition
      \eqref{zok1} holds  $\forall \ p\in
      \Omega$. This proves that the error growth remains bounded in between
      two consecutive events.
    \end{remark}
\vspace{1ex}
  \textbf{Case 2:} Suppose that a DoS attack occurs in the feedback channel at
  the instant $a_{i}\in [0,\ k)$, where $i\in \{0,\ldots, m\}$ represents the
  $i^{\text{th}}$ attack event, and this attack lasts for the duration
  $k_{a_{i}}$ time units. 
  The scalar
    $m$ represents the number of attacks during $[0,\ k)$.  Now, within this
  time interval $k_{a_{i}}$, if an event is not generated then the requirement
  of feedback channel is unnecessary and the measurement error will grow
  according to \eqref{nqdos1}.  The problem is more severe
  if any event occurs within $k_{a_{i}}$ time duration. According
  to~\eqref{nqdos28}, a triggering event occurs only when the stability
  criterion \eqref{nodos3} is violated. Therefore, the unavailability of the
  communication channel may destabilize the system. The effects of this DoS
  attack on the system's
  stability is considered and analyzed in what follows. \\
  Within the interval $[0, k)$, some transmission attempts are not successful
  due to jamming.
  In other words, for the time duration $\displaystyle
  \left(k-{\displaystyle\sum_{i=1}^m} k_{a_{i}}\right)$, the channel is
  available for communication and for the remaining time, the channel is
  unavailable due to the DoS attack. The duration ${\displaystyle\sum_{i=1}^m}
  k_{a_{i}}$ is represented by $T_{\mathrm{off}}(k)$. Then, at $a_{i}$, the
  growth of variable $e$ is
  \begin{eqnarray}\label{nqdos2}
    e(k)=x(k_{i(a_{i})})-x(k),
  \end{eqnarray}
  where $x(k_{i(a_{i})})$ represents the state of the system at the last
  successful control update up to $a_{i}$. At the moment of the attack, the
  condition~\eqref{nodos3} holds. That means
  \[
    \|e(a_{i})\|\leq  \sqrt\mu \|x(a_{i})\|, \quad 
    x(k_{i(a_{i})})-x(a_{i})\leq  \sqrt\mu \|x(a_{i})\|.
  \] 
  Using~\eqref{nqdos2}, the error $e(k)$ can be expressed as
  \begin{eqnarray}\label{nqdos4}
    \|e(k)\|&\leq & (1+\sqrt\mu)\|x(a_{i})\|+\|x(k)\|.
  \end{eqnarray} 
  The inequalities \eqref{nqdos4} and \eqref{nodos3} can be used to compute
  $\Delta V$ as
  \begin{eqnarray}
    \Delta V &\leq & -\xi_{1}\|x(k)\|^{2}+\xi_{2}\left((1+\sqrt\mu)\|x(a_{i})\|+\|x(k)\|\right)^{2}\nonumber\\
    &\leq & 
    ~ \gamma\text{max}\{V(x(k)),V(x(a_{i}))\},\label{nqdos5}
  \end{eqnarray}
  where 
  $\gamma=\frac{\xi_{2}(1+\mu)^{2}}{\lambda_{\min}(P)}$.


Let us consider the $i$th attack interval, i.e. $(a_i,a_i+k_{a_i})$.   Using the  comparison principle for discrete-time system presented in  \cite[Proposition 1]{comp},  for $\tau\in(a_i,a_i+k_{a_i})$, \eqref{nqdos5} reduce to 
\begin{equation}\label{eq:layap with dos}
      V(x(\tau))\leq c_{2}^{(\tau-k_{a_{i}})}V(x({a_{i}})),
\end{equation}
where $c_2 = (1+\gamma)>1$.
Now, consider the consecutive time interval without any DoS attack, i.e. $(a_i+k_{a_i},a_{i+1})$. Again, using comparison principle,  \cite[Proposition 1]{comp}, for $\tau\in(a_i+k_{a_i},a_{i+1})$, \eqref{vb1d} reduces to   
    \begin{equation}\label{eq:lyap without dos}
      V(x(\tau))\leq c_{1}^{(\tau-[a_{i+1}-(a_i+k_{a_i})])}V(x(a_i+k_{a_i})).
  \end{equation}
  Therefore, whenever DoS signal blocks the communication channel, the system dynamics follows \eqref{eq:layap with dos} and in the absence of DoS signal, it is governed by \eqref{eq:lyap without dos}. 
  Recalling that the number of off to on transitions of DoS attack within the interval $[0,k)$ is $N_{\text{off}}(k)$. With these ingredients in mind and combining \eqref{eq:layap with dos} and \eqref{eq:lyap without dos}, we get the following bound on $V(k)$
  \begin{eqnarray}\label{nqdos6}
  V(k)\leq \Xi^{N_{\text{off}}(k)}c_{1}^{(k-T_{\text{off}}(k))}c_{2}^{T_{\text{off}}(k)}V(x(0))
  \end{eqnarray}
  where $\displaystyle \Xi
  ={\lambda_{\text{max}}(P)}/{\lambda_{\text{min}}(P)}$.
  Now, using \eqref{nqdos6} and \eqref{deft1}, we obtain the
  following upper bound for the system's state
  \begin{eqnarray}\label{nqdos7}
    \|x(k)\| \leq  \Xi^{\frac{1+N_{\text{off}}(k)}{2}}c_{1}^{\frac{(k-T_{\mathrm{off}}(k))}{2}}  c_{2}^{\frac{T_{\text{off}}(k)}{2}}\|x(0)\|.
  \end{eqnarray}
  To ensure the convergence of $x(k)$, the following two sub-cases are
  considered.
\begin{description}
\item[ $\Xi=1$:]~ For a selection of $\Xi=1$, inequality \eqref{nqdos7}
    reduces to
    \begin{equation*}
      \|x(k)\|\leq  c_{1}^{\frac{(k-T_{\text{off}}(k))}{2}}\nonumber  c_{2}^{\frac{T_{\text{off}}(k)}{2}}\|x(0)\|.
    \end{equation*}
    Now assume that there exists a scalar $1>\eta_{1}>c_{1}$ such that
    \begin{eqnarray*}
      c_{1}^{\frac{(k-T_{\text{off}}(k))}{2}} c_{2}^{\frac{T_{\mathrm{off}}(k)}{2}}\leq \eta_{1}^{k}.
    \end{eqnarray*}
    After simplification, the following is obtained
    \begin{eqnarray}\label{sdfg}
      \frac{T_{\text{off}}(k)}{k}\leq \frac{2\ln(\eta_{1})-\ln(c_{1})}{\ln(c_{2})-\ln(c_{1})}.
    \end{eqnarray}  
  \item[$\Xi>1$:] \ For a selection of $\Xi>1$, \eqref{nqdos7} reduces to
    \begin{eqnarray}\label{nqdos8}
      \|x(k)\|\leq \Xi^{\frac{1+N_{\text{off}}(k)}{2}}\eta_{1}^{k}\|x(0)\|.
    \end{eqnarray}
    Now, assume that there exists a scalar $1>\eta_{2}>\eta_{1}$ which
    simplifies \eqref{nqdos8} as
    \begin{eqnarray}\label{bve1}
      \Xi^{\frac{1+N_{\text{off}}(k)}{2}}\eta_{1}^{k}\leq \eta_{2}^{k},
    \end{eqnarray}
    and thus $
      \norm{x(k)}\leq \eta_{2}^{k}\norm{x(0)}$.
    The inequity \eqref{bve1} is used to derive the DoS frequency as
    \begin{equation}\label{bfg1}
      \frac{N_{\text{off}}(k)}{k}\leq \frac{2(\ln(\eta_{2})-\ln(\eta_{1}))}{\ln(\Xi)}= {T_a}\nonumber,
    \end{equation}
    where $T_a$ is
    \begin{eqnarray}\label{sdf}
      T_a= \frac{2(\ln(\eta_{2})-\ln(\eta_{1})))}{\ln(\Xi)},
    \end{eqnarray}
which is defined in Assumption \ref{a2s}.
  \end{description}
  From \eqref{nqdos7}, if Assumptions \ref{a1s} and \ref{a2s} hold for the DoS
  attack signal, which are computed from \eqref{sdf} and \eqref{sdfg}, then
  $\|x(k)\|$ in \eqref{nqdos7} is bounded. This completes the proof.
\end{proof}
%
\section{Simulation results}

This section validates the proposed robust control approach in the presence of
DoS attacks with uncertainty in the system's dynamics using a numerical
example.  For the sake of numerical validation, we consider the classical
networked control system corresponding to a batch reactor system~\cite{cvbdf}
with two inputs and two outputs. To realize a stabilizing control law, the
feedback control loop is closed by means of a wireless communication
network. The control input is designed to tackle the aperiodic availability of
feedback information in the presence of  mismatched uncertainty.

\par
We derive a discrete-time linearized model of a batch reactor system in the
form of \eqref{dosa1} from a continuous model with a sampling period
$T=0.05$. The matrices $A$ and $B$ are given by\
\[A=\begin{bmatrix}
  0.0690 &  -0.0100 &   0.3355 & -0.2835\\
  -0.0290  & -0.2145   &   0 &    0.0338\\
  0.0530  &  0.2135 &  -0.3325   & 0.2945\\
  0.0020 & 0.2135 & 0.0670 & 0.1050]
\end{bmatrix},\] and 
\[\ B=\begin{bmatrix}
  0 & 0.2840 & 0.0568 & 0.0568\\
  0  &   0 & -0.1573 & 0
\end{bmatrix}^T.\] 
The matrix $\Delta A$ is defined as $\Delta A=pI$ where variable $p$ is the uncertain parameter with
variations in the unit interval. To design the controller gains, the matrices
$Q=4I$, $R_1=I$, $R_2=I$  and variable $\epsilon=0.01$ are
selected.  The scalar parameter $\sigma$ is chosen to be $0.1 $. The simulation is carried out using MATLAB for a run time of 6
seconds with the initial state $x=\begin{bmatrix} -0.5& -0.3 & 0.2 & -0.05
\end{bmatrix}^{T} $. The matrix $F=2I$ and scalars  $\eta_{1}=0.3$ and
$\eta_{2}=0.95$ are selected such that the conditions \eqref{zok1}, \eqref{scon3}, \eqref{off1}
and \eqref{nqdos26} are satisfied. To obtain the controller gain matrices $K$
and $L$, the Riccati equation \eqref{ri1} is solved leading
to \[K=\begin{bmatrix}
  -0.0710  & -0.9309 &  -0.0356  & -0.1008\\
  1.4597  &  0.1990  &  1.0212  & -0.5773
\end{bmatrix},\] \[\ L=\begin{bmatrix}
  -0.0092 &   0.0057 &  -0.0053  &  0.0092\\
  -0.0144  &  0.0174 &  -0.0072  &  0.0215\\
  0.0220  & -0.0104  &  0.0131 &  -0.0193\\
  0.0007 &  -0.0166  & -0.0016  & -0.0142
\end{bmatrix}.\]


%
Figure~\ref{fi:m1l} shows the convergence of the state $x$ in spite of
system's uncertainty and DoS attack on the communication channel. The attack
signal is represented with the red color. The degradation of the system's
performance following DoS attacks is apparent in Fig.~\ref{fi:m1l}. 
\begin{figure}[t]
  \begin{center}
    \includegraphics[width=\columnwidth]{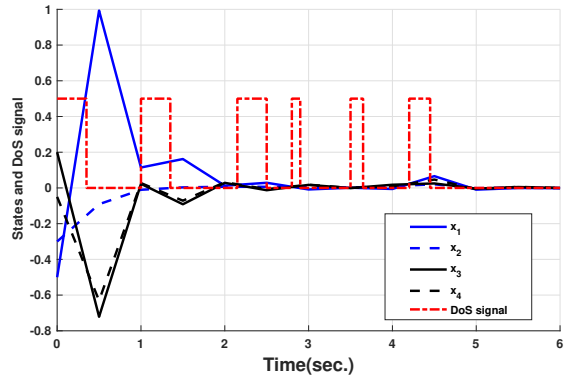}
    \caption{ Convergence of states in the presence of DoS attacks for
      $p=0.5$.}\label{fi:m1l}
  \end{center}
\end{figure}
\begin{table}[htbp]
\caption {Comparison of event-triggered vs. periodic feedback control} 
  \label{lch5tdfg}
  \begin{center}
    \begin{tabular}{lllll}
      \toprule 
      Control Strategy & $\tau_{\text{max}}$(sec.) & $\tau_{\text{min}}$(sec.) &  $u_{\text{total}}$ \tabularnewline  \midrule\midrule
      Periodic feedback control & $0.05$ & $0.05$  & $120$ \tabularnewline\midrule
      \pbox{20 cm}{ Event-triggered control with DoS } & $ 0.93$ & $0.05$ &$37$ \tabularnewline  \bottomrule
    \end{tabular}
  \end{center}
\end{table}
Table~\ref{lch5tdfg} shows the efficacy of the proposed control
algorithm. The symbol $u_{\text{total}}$ denotes the total number of
transmissions of control inputs via the communication network. The quantities
$\tau_{\text{min}}$ and $\tau_{\text{max}}$ represent the minimum and maximum
duration of inter-event time respectively. The larger inter-event time,
the improved savings in communication resources. The lower bound of attack duration, total DoS  period and frequency are computed as $T_{\mathrm{a}}=0.1$ sec., $T_{\mathrm{off}}=1.53$ sec., 
$N_{\mathrm{off}}=12$. To generate the DoS signal we have used these bounds.

%
\section{Conclusion}
  

In this paper, we investigated the robust stabilization of discrete-time
mismatched uncertain systems in the presence of DoS attack. The primary
contribution of this paper is an explicit characterization of the attack
signal, namely DoS duration and frequency under which the mismatched system
remains input-to-state stable with event-triggered feedback. The aperiodic use
of feedback information significantly reduces the communication overhead over
the transmission network. To handle the inherent uncertainty in the system's
model, an optimal control approach based on a robust control technique has
been considered. The proposed robust control approach translates the robust
control problem into an optimal control one for a virtual system with a
modified cost-functional. The optimal input for the virtual system is the
robust solution for uncertain system. The proposed robust controller also
ensures the stability of closed-loop system under a generic class of DoS
attacks, for which the attack signal satisfies Assumptions \ref{a1s} and
\ref{a2s}.  Beyond its effectiveness in overcoming the damaging effects of DoS
attacks, the developed event-triggered control technique leads to significant
savings in the channel bandwidth. The proposed control algorithm is
illustrated and validated numerically using the classical NCS batch reactor
model.

%


\end{document}